\documentclass[11pt, a4paper]{article}

\usepackage{setspace}
\usepackage{bm}
\usepackage{upgreek}
\usepackage{cite}
\usepackage{rotating}
\usepackage{graphicx}
\usepackage{verbatim}
\usepackage{wrapfig}
\usepackage[margin=2.5cm]{geometry}
\usepackage{amsmath}
\usepackage{placeins}
\usepackage{url}

\usepackage{fancyhdr}
\usepackage{color}
\usepackage{amsthm}
\usepackage{amsmath, amssymb}
\usepackage{mathabx}
\usepackage{longtable}
\usepackage{lscape}
\usepackage{multicol}
\usepackage{pdfpages}
\usepackage{siunitx}
\usepackage{graphicx}
\usepackage{enumerate}

\usepackage[affil-it]{authblk}

\usepackage{caption}
\usepackage{subcaption}
\usepackage{bbm}
\usepackage{nameref}

\usepackage{tikz}
\usetikzlibrary{positioning}
\tikzset{>=stealth}
\usetikzlibrary{decorations.markings}
\usetikzlibrary{patterns,angles,quotes}
\usetikzlibrary{math}
\usetikzlibrary{calc}
\usetikzlibrary{shapes,backgrounds}
\usetikzlibrary{arrows.meta,intersections}

\usepackage{hyperref}

\usepackage{chngcntr}
\counterwithout{table}{section}

\newtheorem{all}{Theorem}[section]

\theoremstyle{plain}
\newtheorem{lemma}[all]{Lemma}
\newtheorem{thm}[all]{Theorem}

\theoremstyle{definition}

\newtheorem{rem}[all]{Remark}
\newtheorem{asptn}[all]{Assumption}

\newcommand{\BC}{{\mathbb{C}}}

\newcommand{\BI}{{\mathbb{I}}}

\newcommand{\BR}{{\mathbb{R}}}
\newcommand{\BS}{{\mathbb{S}}}

\newcommand{\FT}{{\mathcal{F}}}

\newcommand{\dd}{{\mathrm{d}}}

\newcommand{\sgn}{\mathrm{sgn}}

\newcommand{\tr}{\mathrm{Tr}\ }

\newcommand{\com}[1]{}

\makeatletter
\pgfdeclarepatternformonly[\LineSpace]{my north east lines}{\pgfqpoint{-1pt}{-1pt}}{\pgfqpoint{\LineSpace}{\LineSpace}}{\pgfqpoint{\LineSpace}{\LineSpace}}%
{
    \pgfsetcolor{\tikz@pattern@color}
    \pgfsetlinewidth{0.4pt}
    \pgfpathmoveto{\pgfqpoint{0pt}{0pt}}
    \pgfpathlineto{\pgfqpoint{\LineSpace + 0.1pt}{\LineSpace + 0.1pt}}
    \pgfusepath{stroke}
}

\pgfdeclarepatternformonly[\LineSpace]{my north west lines}{\pgfqpoint{-1pt}{-1pt}}{\pgfqpoint{\LineSpace}{\LineSpace}}{\pgfqpoint{\LineSpace}{\LineSpace}}%
{
    \pgfsetcolor{\tikz@pattern@color}
    \pgfsetlinewidth{0.4pt}
    \pgfpathmoveto{\pgfqpoint{0pt}{\LineSpace}}
    \pgfpathlineto{\pgfqpoint{\LineSpace + 0.1pt}{-0.1pt}}
    \pgfusepath{stroke}
}
\makeatother

\newdimen\LineSpace
\tikzset{
    line space/.code={\LineSpace=#1},
    line space=17pt
}

\DeclareMathOperator\artanh{artanh}

\numberwithin{equation}{section}
\begin{document}
\title{Linear Criterion for an Upper Bound on the Bardeen-Cooper-Schrieffer Critical Temperature}
\author[1]{Barbara Roos\footnote{Author to whom correspondence should be addressed: barbara.roos@uni-tuebingen.de}}
\affil[1]{University of T\"ubingen, Fachbereich Mathematik, Auf der Morgenstelle 10, 72076 T\"ubingen, Germany}

\date{\today}

\maketitle

\begin{abstract}
Since Bardeen-Cooper-Schrieffer theory of superconductivity is non-linear, it is difficult to study superconducting properties analytically.
There is a more tractable linear criterion which determines a temperature $T_l$ below which the system is superconducting.
Here, we observe that there is a similar linear criterion which gives a temperature $T_u$ above which no superconductivity occurs.
We provide examples of translation invariant systems where $T_u>T_l$ as well as systems where $T_u=T_l$.
Furthermore, we estimate $T_u$ for half-spaces and show that it is exponentially small in the weak coupling limit, exhibiting the same asymptotics as the critical temperature for full space.
\end{abstract}

\section{Introduction}
Bardeen-Cooper-Schrieffer (BCS) theory is a successful model for superconductivity.
Recent developments include the rigorous derivation of the Ginzburg-Landau equation from BCS theory \cite{frank_microscopic_2012,frank_external_2016,deuchert_microscopic_2023, deuchert_microscopic_2023-1}, universality of the ratio of the energy gap and the critical temperature \cite{langmann_universal_2023-1,henheik_universal_2023, henheik_universality_2023}, and the existence of boundary superconductivity \cite{hainzl_boundary_2022, roos_bcs_2023, roos_enhanced_2023}.
While BCS theory is difficult to study in general due to its non-linearity, there is a linear criterion to derive a temperature $T_l$ below which the system is superconducting \cite{hainzl_bcs_2008-1,cadamuro_bcs_2018}.
This linear criterion was for instance applied to study boundary superconductivity in \cite{hainzl_boundary_2022, roos_bcs_2023, roos_enhanced_2023}.
It is expected that superconductivity can occur above $T_l$ in some systems \cite{caldas_temperature_2008,freiji_gap_2012, rodriguez_salmon_first-order_2019}. 
In this paper we observe that there is a linear criterion, which gives a temperature $T_u\geq T_l$ above which no superconductivity occurs.
If the two temperatures agree, $T_c:=T_u=T_l$ is the unique critical temperature of the system which separates the superconducting and normal phase.

We compare $T_u$ and $T_l$ for translation-invariant systems, where we restrict to $SU(2)$-invariant states, but allow for nonzero total momentum.
It has been known that if the minimizer of the BCS functional has zero total momentum, i.e.~is translation invariant, there is a unique critical temperature $T_c$ \cite{hainzl_bardeencooperschrieffer_2016}.
However, it is unclear whether the minimizer indeed is translation invariant.
There is an argument in \cite{hainzl_bardeencooperschrieffer_2016} which shows translation invariance of the minimizer if the BCS functional is minimized over a larger set of states, where the Cooper-pair wave function is not required to be symmetric.
It was observed in \cite{deuchert_persistence_2018} that the argument can be adapted to the setting with symmetry for suitable interactions.
We discuss the precise conditions under which the argument shows that $T_u$ and $T_l$ agree, and show that there are translation invariant systems where $T_u>T_l$.
In particular, we point out that it is an open question whether there is a unique critical temperature in all translation-invariant BCS functionals restricted to $SU(2)$-invariant states.

Furthermore, we study $T_u$ on half-spaces and derive the first rigorous upper bounds on the BCS critical temperatures for systems with boundary.
We extend results from \cite{hainzl_boundary_2022, roos_bcs_2023} on the relative difference of the critical temperatures on half- and full space in the weak and strong coupling limits.
It turns out that even though superconductivity occurs at elevated temperatures in the presence of a boundary \cite{hainzl_boundary_2022, roos_bcs_2023}, the effect vanishes in the weak-coupling limit.
In particular, on half-spaces superconductivity can only occur at temperatures exponentially small in the coupling, like on full space.

In the following, we explain the two linear criteria in more detail.
In Section~\ref{sec:results} we present and discuss our main results.

\subsection{The two linear criteria}\label{sec:criteria}
BCS theory is based on minimizing the BCS functional.
If the minimizer has a non-trivial pairing term $\alpha$, the system is in a superconducting state.
Given a system with shape $\Omega$ and assuming $SU(2)$-invariance, the functional is defined on Hermitian states $\Gamma$ of the following form
\begin{equation}
\Gamma=\begin{pmatrix} \gamma & \alpha\\ \overline{\alpha} & 1- \overline{\gamma} \end{pmatrix},
\end{equation}
where $\gamma$ and $\alpha$ are operators acting on $L^2(\Omega)$ and $0\leq \Gamma \leq 1$.
The operator $\gamma$ is self-adjoint, while $\alpha$ is symmetric, i.e. the kernel satisfies $\alpha(x,y)=\alpha(y,x)$.
The BCS functional is given by
\begin{equation}
\mathcal{F}(\Gamma)=\tr(h\gamma) -TS(\Gamma)- \int_{\Omega\times \Omega} V(x-y) \vert \alpha(x,y)\vert^2 \dd  x\dd y ,
\end{equation}
where $h$ is a one-particle Hamiltonian,  $T$ is the temperature, $S(\Gamma)=-\tr(\Gamma \ln \Gamma) $ is the entropy,  $V$ is the effective interaction between the electrons, and $\alpha(x,y)$ is the integral kernel of $\alpha$ \cite{hainzl_bardeencooperschrieffer_2016}.
Compared to \cite{hainzl_bardeencooperschrieffer_2016} the chemical potential has been absorbed into the one-particle Hamiltonian and does not appear explicitly in the functional.

The corresponding Euler-Lagrange equation is given by \cite{frank_bcs_2019}
\begin{equation}\label{el-eq}
H_\Delta+T \ln \left(\frac{\Gamma}{1-\Gamma}\right)=0
\end{equation}
where
\begin{equation}
H_\Delta=\begin{pmatrix} h & \Delta \\ \overline{\Delta} & -\overline{h} \end{pmatrix}
\end{equation}
and $\Delta(x,y)=-2V(x-y)\alpha(x,y)$.
There is always one solution of the Euler-Lagrange with $\alpha=0$, called the normal state.
The corresponding $\gamma$ is given by $\gamma=(1+\exp(h/T))^{-1}$.
The question is, whether the normal state minimizes the BCS functional.
If the Hessian of the functional at the normal state has negative spectrum, the normal state is not the minimizer and therefore there is superconductivity.
To express the Hessian, define the operator
\begin{equation}\label{eq:KT}
L_{T}=\frac{h_x+h_y}{\tanh(h_x/2T)+\tanh(h_y/2T)}
\end{equation}
on $L^2(\Omega\times \Omega)$ through functional calculus \cite{frank_bcs_2019}. 
Let $V$ act on $L^2(\Omega\times \Omega)$ as multiplication with $V(x-y)$.
If one computes the Hessian for variations of $\alpha$, one finds that it equals $L_{T} -V$ restricted to symmetric functions $\alpha(x,y)=\alpha(y,x)$.
If this operator has spectrum below zero, there is superconductivity.
One can define
\begin{equation}
T_l:=\inf\{ T>0: \inf \sigma_s(L_T-V) \geq 0\},
\end{equation}
where $\sigma$ denotes the spectrum and the subscript $s$ denotes that the domain is restricted to symmetric functions.
The system is in a superconducting state for $T<T_l$.

We now observe that a similar linear criterion gives a temperature $T_u$ such that the system is in the normal state for $T>T_u$.
The idea of having a linear criterion for an upper bound on the critical temperature has previously occurred in \cite{freiji_gap_2012} for a translation invariant, spin polarized system.
We need the following two Lemmas, which are taken from \cite[Lemma 4.1, Lemma 4.2]{hainzl_bardeencooperschrieffer_2016}.
These results go back to \cite{frank_microscopic_2012} where they were essential for the derivation of the Ginzburg-Landau functional, and a similar result was derived earlier in \cite{hainzl_nonlinear_2008}.
\begin{lemma}\label{lemma1}
Let $\Gamma_0=(1+e^{H_0/T})^{-1}$ be the normal state.
With the relative entropy
\begin{equation}
D(\Gamma||\Gamma'):=\frac{1}{2}\tr [ \Gamma(\ln \Gamma-\ln \Gamma')+(1-\Gamma)(\ln(1-\Gamma)-\ln(1-\Gamma'))]
\end{equation}
one can write
\begin{equation}
\mathcal{F}(\Gamma)-\mathcal{F}(\Gamma_0)=T D(\Gamma||\Gamma_0)-\int_{\Omega\times \Omega} V(x-y)\vert \alpha(x,y)\vert^2 \dd x \dd y
\end{equation}
\end{lemma}
\begin{lemma}\label{lemma2}
With $\Gamma_0=(1+e^{H_0/T})^{-1}$ we have
\begin{equation}
T D(\Gamma||\Gamma_0)\geq \frac{1}{2} \tr \left[ \frac{H_0}{\tanh \frac{H_0}{2T} }(\Gamma-\Gamma_0)^2\right]+\frac{2}{3}T\tr[(\Gamma(1-\Gamma) - \Gamma_0(1-\Gamma_0))^2]
\end{equation}
\end{lemma}
Neglecting nonnegative terms, one can bound
\[
T D(\Gamma||\Gamma_0)\geq \frac{1}{2}\tr \left[ \frac{H_0}{\tanh \frac{H_0}{2T}}\begin{pmatrix}\alpha \overline{\alpha}&0\\0& \overline{\alpha}\alpha\end{pmatrix}\right].
\]
The expression on the right equals $\langle \alpha, D_T \alpha\rangle$, where
\begin{equation}\label{eq:DT}
D_{T}=\frac{1}{2} \frac{h_x}{\tanh(h_x/2T)}+\frac{1}{2} \frac{h_y}{\tanh(h_y/2T)}
\end{equation}
is defined in $L^2(\Omega\times \Omega)$. 
Combining this with Lemma~\ref{lemma1}, gives
\begin{equation}
\mathcal{F}(\Gamma)-\mathcal{F}(\Gamma_0)\geq \langle \alpha, (D_T-V)\alpha \rangle.
\end{equation}
In particular, if $\inf \sigma_s(D_T-V)>0$, the system is in the normal state.
Define $T_u$ as
\begin{equation}\label{eq:Tu}
T_u:=\sup\{ T> 0: \inf \sigma_s(D_T-V) \leq 0\}.
\end{equation}
The system is in the normal state for temperatures above $T_u$.

\subsection{Results}\label{sec:results}
We now fix the one-particle Hamiltonian to be $h=-\Delta-\mu$, where $\Delta$ is the Laplace operator with appropriate boundary conditions and $\mu\in \BR$ is the chemical potential and consider dimensions $d\in\{1,2,3\}$.
In the translation invariant case, we assume $\Omega$ to be $\BR^d$.
The other domains $\Omega$ we shall consider are half-spaces $(0,\infty)\times \BR^{d-1}$ with Dirichlet or Neumann boundary conditions.
For the interaction $V$ we shall assume the following.
\begin{asptn}\label{as1}
Let $d\in\{1,2,3\}$. Assume that
\begin{enumerate}[(i)]
\item $V\in L^{p_d}(\BR^d)$, where $p_1=1$, $p_2>1$, and $p_3\geq 3/2$. For $d=1$, $V$ may also be the difference of two bounded, positive Borel measures.
\item $V(r)=V(-r)$
\end{enumerate}
 \end{asptn}
With these assumptions on $V$, $D_T-V$ and $L_T-V$ define self-adjoint operators via the KLMN theorem, see e.g.~\cite[Section 11.3]{lieb_analysis_2001} and \cite[Theorem 6.24]{teschl_mathematical_2014}.
Both $\inf \sigma_s(L_T-V)$ and $\inf \sigma_s(D_T-V)$ are continuous in $T$, the proof is analogous to  \cite[Lemma 4.1]{roos_bcs_2023}.
Furthermore, since $L_T\geq D_T\geq 2T$ \cite[Lemma 2.7]{roos_bcs_2023}, the temperatures $T_l$ and $T_u$ satisfy $T_l\leq T_u<\infty$.
Since $D_T$ is strictly increasing in $T$, if $T_u$ is positive, it is the unique temperature satisfying $\inf \sigma_s(D_T-V)=0$.

\subsubsection{Translation invariant systems}
It has been shown in \cite{hainzl_bcs_2008-1}, that if one restricts the BCS functional to translation invariant states $\alpha(x,y)=\alpha(x-y)$, there is a unique critical temperature separating the normal from the superconducting phase.
Define the operator $K_T-V$ acting in $L^2(\BR^d)$, where
\begin{equation}
K_T = \frac{h}{\tanh \frac{h}{2T}}
\end{equation}
and with slight abuse of notation $V$ here denotes the one-body operator acting as multiplication with $V(r)$.
We shall see that the restriction of both operators $L_T-V$ and $D_T-V$ to such states with zero total momentum give $K_T-V$.
Thus, the lower and upper bounds on the critical temperature match in this case.
Now, one may ask whether the minimizer of the BCS functional is translation invariant.
This question has been addressed in \cite[Section F]{hainzl_bardeencooperschrieffer_2016}.
The arguments presented there show that the BCS functional has a translation invariant minimizer on an enlarged domain, where the symmetry condition $\alpha(x,y)=\alpha(y,x)$ is dropped.
We follow the arguments in \cite[Section F]{hainzl_bardeencooperschrieffer_2016} to derive sufficient conditions for $T_l$ and $T_u$ to agree when the symmetry condition $\alpha(x,y)=\alpha(y,x)$ is imposed.

Note that since $V(r)=V(-r)$, the operator $K_T-V$ commutes with reflections $\psi(r)\mapsto \psi(-r)$.
Let $\sigma_{s/a} (K_T-V)$ denote the spectrum of $K_T-V$ restricted to even/odd functions.
In momentum space, $K_T$ acts as multiplication by
\begin{equation}\label{KT}
K_T(p)=\frac{p^2-\mu}{\tanh \frac{p^2-\mu}{2T}}
\end{equation}
which attains the minimal value $2T$ on the Fermi sphere $p^2=\mu$ for $\mu>0$.

\begin{lemma}\label{lea:uniquetc}
Let $d\in\{1,2,3\}$ and $\mu\in \BR$.
Let $\Omega=\BR^d$ and let $V$ satisfy Assumption~\ref{as1}.
If there is a temperature $T_c$ such that $\inf \sigma_a (K_{T_c}-V)\geq\inf \sigma_s(K_{T_c}-V) =0$, then $T_c=T_u=T_l$.
\end{lemma}

\begin{rem}\label{rem:uniquetc}
The conditions in Lemma~\ref{lea:uniquetc} are satisfied in several situations.
The first example are interactions $V\in L^1$ with nonnegative Fourier transform $\widehat{V}(p)=(2\pi)^{-d/2}\int_{\BR^d} e^{-i p\cdot r} V(r) \dd r \geq 0$ and $\widehat{V}(0)>0$ \cite{deuchert_persistence_2018}.
To see this, we use analogously to \cite{hainzl_bardeencooperschrieffer_2016} that $\widehat{V}\geq 0$ implies that $\langle \widehat \psi, \widehat{V} * \widehat \psi  \rangle \leq \langle |\widehat \psi|, \widehat{V} * |\widehat \psi |\rangle$. 
Now the point is that if $\psi$ is odd, then $|\widehat \psi|$ is even and hence $\inf \sigma_a (K_{T}-V)\geq \inf \sigma_s(K_{T}-V) $.
Furthermore, if $\mu>0$ and $\widehat V(0)>0$, there is a unique $T_c>0$ such that $\sigma_s(K_{T_c}-V)=0 $ \cite{frank_critical_2007}.

There are more general examples in $d\in\{1,2,3\}$ at weak enough coupling and $\mu>0$.
Consider interactions $\lambda V$ with $V\in L^1$.
The weak coupling properties are encoded in the operator $\mathcal{V}_\mu:L^2(\BS^{d-1})\to L^2(\BS^{d-1})$,
with integral kernel
\begin{equation}\label{eq:Vmu}
\mathcal{V}_\mu(p,q)=\frac{1}{(2\pi)^{d/2}} \widehat V(\sqrt{\mu}(p-q)).
\end{equation}
Let $e_\mu^{s/a}:=\sup \sigma_{s/a} \mathcal{V}_\mu$.
Note that for $d=1$ we have $e_\mu^{s/a}=\frac{\widehat V(0) \pm \widehat V(2\sqrt{\mu})}{(2\pi)^{1/2}}$, where $s$ corresponds to the $+$ and $a$ to the $-$ sign, respectively.
If $e_\mu^s>0$, for all $\lambda>0$ there is a temperature $T_c(\lambda)>0$ such that $\inf \sigma_s(K_{T_c(\lambda)}-V) =0$ \cite{frank_critical_2007, henheik_universality_2023}.
Furthermore, if $e_\mu^s >e_\mu^a$, then the ground state of $K_{T_c(\lambda)}-\lambda V$ is even at weak enough coupling.
This follows from the asymptotics of the corresponding Birman-Schwinger operators proved in \cite{frank_critical_2007, henheik_universality_2023}.
Therefore, the conditions in Lemma~\ref{lea:uniquetc} to have a unique critical temperature are satisfied at weak enough coupling if $e_\mu^s>\max\{e_\mu^a,0\}$.
The condition $e_\mu^s>\max\{e_\mu^a,0\}$ is satisfied in dimension one if $\widehat{V}(0)$ and $\widehat{V}(2\sqrt{\mu})$ are positive.
In dimensions two and three, the condition is satisfied for instance if $\widehat V\geq 0, \widehat V\not \equiv 0$ on the ball $|p|<2\sqrt{\mu}$.
\end{rem}
\begin{rem}
The argument that we use to prove Lemma~\ref{lea:uniquetc} was implicitly used in \cite{deuchert_persistence_2018} to show that for translation invariant systems the minimizer of the BCS functional is translation invariant in a temperature interval below $T_c$ for radial interactions.
The main results \cite[Theorem 2.1 and Theorem 2.8]{deuchert_persistence_2018} effectively assume that $\inf \sigma(K_{T_c}-V)= \inf \sigma_s(K_{T_c}-V)=0$, which is the setting where Lemma~\ref{lea:uniquetc} holds.
However, they do not clearly state that this is an additional assumption on $V$.
\end{rem}

While we have now seen cases where $T_u=T_l$, the following Theorem provides an example where $T_u>T_l$.
Let $e_0^s=2\frac{\widehat V(0)}{(2\pi)^{1/2}}$.
\begin{thm}\label{thm:tu-weak-coupling}
Let $\mu>0$.
Let $V$ satisfy $V(r)=V(-r)$ as well as $(1+|\cdot|^2)V \in L^1(\BR)$ and $e_\mu^s>0$.
Then for all $\lambda>0$ the temperatures $T_l(\lambda)$ and $T_u(\lambda)$ are positive.
As $\lambda\to 0$ the asymptotics of $T_l(\lambda)$ and $T_u(\lambda)$ are given by
\begin{align}
T_l(\lambda)&=\mu e^{-\mu^{1/2}/(\lambda e_\mu^s ) +O(1)},\\
T_u(\lambda)&=\mu e^{-\mu^{1/2}/(\lambda \max\{e_\mu^s,\frac{1}{2}e_0^s \}) +O(1)}.
\end{align}
In particular, if $\hat V(0)>0$ and $\hat V(2\sqrt{\mu})< 0$,  the ratio $T_u(\lambda)/T_l(\lambda) \to \infty$ in the limit $\lambda \to 0$.
\end{thm}
This theorem shows, that with the existing methods it is not possible to determine whether for general $V$ there is a unique critical temperature for the BCS functional, even in a translation-invariant setting.
It is thus an open problem to clarify, whether there is a unique critical temperature, and, if yes, where in the interval $[T_l,T_u]$ it lies.

\begin{rem}
Considering a setting analogous to Theorem~\ref{thm:tu-weak-coupling} in dimensions two and three, we expect $T_l(\lambda)$ and $T_u(\lambda)$ both to be given by
\begin{equation}\label{tc:asy-23}
\mu e^{-\mu^{1-d/2}/(\lambda e_\mu^s ) +O(1)}
\end{equation}
for $\lambda\to 0$.
In particular, finding suitable $V$ where $T_l(\lambda)<T_u(\lambda)$ for $d\in\{2,3\}$, would require a more involved analysis going to higher orders of the expansion.

Let us explain the intuition behind \eqref{tc:asy-23}.
First, notice that this is the leading order asymptotics of the critical temperature defined using $K_T-V$, computed in \cite{hainzl_bardeencooperschrieffer_2016,henheik_universality_2023}.
In dimension one, for both $T_l$ and $T_u$ the term $ e_\mu^s $ originates in the singularity of the respective Birman-Schwinger operator at total momentum zero.
Additionally, for $T_u$ the contribution $\frac{1}{2}e_0^s $ stems from a singularity of the Birman-Schwinger operator at total momentum $\sqrt{\mu}$, as explained in Sections~\ref{sec:wc} and \ref{sec:pf2}.
In dimensions two and three, however, the Birman-Schwinger operator only has a singularity at total momentum zero (see Section~\ref{sec:pf3}).
This is why we only expect the contribution $e_\mu^s$ from total momentum zero to appear.
Since at total momentum zero, $L_T-V$ and $D_T-V$ agree with $K_T-V$, we expect the asymptotics of $T_l$ and $T_u$ to agree to leading order with the asymptotics derived in \cite{hainzl_bardeencooperschrieffer_2016,henheik_universality_2023}.

\end{rem}

\begin{rem}
Consider $\mu>0$ and an interaction $\lambda V$ in 1d with $e_\mu^s >0$ and $\lambda>0$.
In order for $T_l(\lambda)$ and $T_u(\lambda)$ to agree at weak enough coupling, $e_\mu^s> e_\mu^a$ or equivalently $\widehat V(2\sqrt{\mu})>0$, is a sufficient condition according to Remark~\ref{rem:uniquetc}.
On the other hand, Theorem~\ref{thm:tu-weak-coupling} implies that $e_\mu^s\geq e_\mu^a$, or $\widehat V(2\sqrt{\mu})\geq 0$, is a necessary condition.
\end{rem}

\subsubsection{Half-spaces}
Assume that the interaction $\lambda V$ is such that on $\Omega_0=\BR^d$ at weak enough coupling there is a unique critical temperature $T_l^0(\lambda)=T_u^0(\lambda)=T_c^0(\lambda)$ determined through $\inf \sigma_s (K_{T_c^0(\lambda)}-\lambda V)=0$.
We introduce a Dirichlet or Neumann boundary cutting space in half and consider $\Omega_1=(0,\infty)\times \BR^{d-1}$ and the corresponding $T_l^1(\lambda)$ and $T_u^1(\lambda)$.
It was shown in \cite{hainzl_boundary_2022, roos_bcs_2023} that there are systems where $T_l^1(\lambda)>T_c^0(\lambda)$, i.e.~where in the presence of a boundary, superconductivity persists at higher temperatures.
Furthermore, it was shown that the relative difference of $T_l^1(\lambda)$ and $T_c^0(\lambda)$ vanishes in the weak coupling limit.
Here we extend these results and show that the relative difference of $T_u^1(\lambda)$ and $T_c^0(\lambda)$ vanishes in the weak coupling limit.
This bounds the temperature range at which boundary superconductivity can occur.

\begin{thm}\label{thm:rel-tu-difference}
Let $\mu>0$ and $d\in\{1,2,3\}$.
Let $V$ satisfy Assumption~\ref{as1}, $V\geq 0$,  $(1+|\cdot|^2)V \in L^1(\BR)$ for $d=1$ and $V \in L^1(\BR^d)$ for $d\in\{2,3\}$, and $e_\mu^s>\max\{0,e_\mu^a\}$.
Then
\begin{equation}
\lim_{\lambda\to 0} \frac{T_u^1(\lambda)-T_c^0(\lambda)}{T_c^0(\lambda)}=0
\end{equation}
\end{thm}
This implies that asymptotically $T_u^1(\lambda)=T_c^0(\lambda)(1+o(1))$.
For the latter, the asymptotics were computed in \cite{hainzl_bardeencooperschrieffer_2016, henheik_universality_2023}, which to leading order read
\[
T_u^1(\lambda)=T_c^0(\lambda)(1+o(1))= \mu \exp\left(-\frac{\mu^{1-d/2}}{\lambda  e_\mu^s} +O(1) \right).
\]
In particular, also on half-spaces superconductivity only occurs at temperatures exponentially small as $\lambda \to 0$.
\begin{rem}
We believe that the assumption $V\geq0$ in Theorem~\ref{thm:rel-tu-difference} is not necessary for the result to hold.
However, it is required by our proof method.
\end{rem}

In the last result we consider the strong coupling limit.
Recall from Remark~\ref{rem:uniquetc} that for the system on $\BR$ with interaction $V=\lambda \delta$, for all $\lambda>0$ there is a unique critical temperature $T_c^0(\lambda)=T_l^0(\lambda)=T_u^0(\lambda)$.
In \cite{hainzl_boundary_2022} it was shown that for Dirichlet boundary conditions on the half-line, the temperature $T_l^1(\lambda)$ satisfies
\begin{equation}
\lim_{\lambda\to \infty} \frac{T_l^1(\lambda)-T_c^0(\lambda)}{T_c^0(\lambda)}=0.
\end{equation}
Here we improve this result and show that the relative difference of $T_u^1(\lambda)$ and $T_c^0(\lambda)$ vanishes in the strong coupling limit.

\begin{thm}\label{thm:rel-tu-difference-strongc}
Let $\mu>0$, $d=1$ and let the interaction be $V=\lambda \delta$ for $\lambda>0$.
Assume Dirichlet boundary conditions on $(0,\infty)$.
In the strong coupling limit
\begin{equation}
\lim_{\lambda\to \infty} \frac{T_u^1(\lambda)-T_c^0(\lambda)}{T_c^0(\lambda)}=0.
\end{equation}
\end{thm}

The remainder of the paper is structured as follows.
In Section~\ref{sec:rel} we compare the operators $L_T, D_T$ and $K_T$ and prove Lemma~\ref{lea:uniquetc}.
In Section~\ref{sec:wc} we analyze the weak coupling asymptotics of the Birman-Schwinger operators corresponding to $D_T-\lambda V$ and $L_T- \lambda V$.
These asymptotics are then used to prove Theorems~\ref{thm:tu-weak-coupling} and \ref{thm:rel-tu-difference} in Sections~\ref{sec:pf2} and \ref{sec:pf3}, respectively.
The proof of Theorem~\ref{thm:rel-tu-difference-strongc} is discussed in Section~\ref{sec:pf4}.

\section{Relation of $L_T$ and $D_T$ to $K_T$}\label{sec:rel}
Assume that $\Omega=\BR^d$ in this section.
We claimed in the introduction that $L_T-V$ and $D_T-V$ restricted to zero total momentum agree with $K_T-V$.
The operator $K_T$ acts as multiplication with
\[
K_T(p)=\frac{p^2-\mu}{\tanh\Big( \frac{p^2-\mu}{2T}\Big)}
\]
in momentum space.
Let $U$ be the unitary operator on $L^2(\BR^d\times \BR^d)$ switching to relative and center of mass coordinates $r=x-y$, $z=x+y$,
\begin{equation}\label{def:U}
 U\psi(r,z)=\frac{1}{2^{d/2}} \psi((r+z)/2,(z-r)/2).
\end{equation}
The operator $U (D_T-V) U^\dagger$ acts as
\[
U (D_T-V) U^\dagger \psi (r,z)= \int_{\BR^{4d}}  \frac{e^{i p \cdot (r-r') + i q \cdot( z-z')}}{(2\pi)^{2d}} N_T(p,q)^{-1}  \psi(r',z') \dd r' \dd z' \dd p \dd q  - V(r) \psi(r,z),
\]
where
\begin{equation}\label{NT}
N_T(p,q)=2 \left(\frac{(p+q)^2-\mu}{\tanh\left(\frac{(p+q)^2-\mu}{2T}\right)}+\frac{(p-q)^2-\mu}{\tanh\left(\frac{(p-q)^2-\mu}{2T}\right)}\right)^{-1}.
\end{equation}
Similarly, for $U (L_T-V) U^\dagger$ replace $N_T$ by $B_T$ in the equation above, where $B_T$ is defined as 
\begin{equation}\label{BT}
B_T(p,q)=\frac{1}{2} \frac{\tanh\left(\frac{(p+q)^2-\mu}{2T}\right)+\tanh\left(\frac{(p-q)^2-\mu}{2T}\right)}{p^2+q^2-\mu}.
\end{equation}
Note that $B_T(p,q)\leq N_T(p,q)$ and thus $L_T\geq D_T$, see \cite[Lemma 2.7]{roos_bcs_2023}. 
This inequality is slightly stronger than a similar inequality for Birman-Schwinger operators used in \cite[Proposition 23]{frank_bcs_2019} and \cite[Proposition 17]{cadamuro_bcs_2018}.
Both operators $U (D_T-V) U^\dagger$ and $U (L_T-V) U^\dagger$ are translation invariant in $z$-direction and we can restrict them to a fixed total momentum $q$.
Since $N_T(p,0)^{-1}=B_T(p,0)^{-1}=K_T(p)$, indeed $K_T-V$ is the restriction of $L_T-V$ and $D_T-V$ to zero total momentum.

\subsection{Proof of Lemma~\ref{lea:uniquetc}}\label{sec:pf1}
We are now ready to prove Lemma~\ref{lea:uniquetc}.
The proof is an adaptation of the arguments in \cite[Section F]{hainzl_bardeencooperschrieffer_2016} and \cite[Lemma 2.4]{roos_bcs_2023}.
\begin{proof}[Proof of Lemma~\ref{lea:uniquetc}]
We want to argue that $\inf \sigma_s (L_{T_c}-V)=\inf \sigma_s(D_{T_c}-V)=0$.
For any $\alpha\in H^1(\BR^d\times \BR^d)$ with $\alpha(x,y)=\alpha(y,x)$, let $\tilde \alpha(r,y):= \alpha(r+y,y)$.
Since $\alpha$ is symmetric, we may write
\begin{equation}
\langle \alpha, (D_T-V) \alpha \rangle
=\int_{\BR^d} \langle  \alpha(\cdot,y),  K_T  \alpha(\cdot, y )\rangle \dd y
-\int_{\BR^{2d}} V(x-y)\vert  \alpha(x,y)\vert^2 \dd x \dd y\\
\end{equation}
Using the translation invariance of $K_T$ in the first term and substituting $r=x-y$ in the second term, we obtain
\begin{equation}
\langle \alpha, (D_T-V) \alpha \rangle
=\int_{\BR^d} \langle \tilde \alpha(\cdot,y),  (K_T-V )\tilde \alpha(\cdot, y )\rangle \dd y
\end{equation}
Note that $\tilde \alpha(r,y)$ is not necessarily even in $r$.
Therefore, $\inf \sigma_s(D_T-V)\geq \inf \sigma (K_T-V)$ without the restriction to symmetric states.
At the same time,  $\inf \sigma_{s}( K_T-V )\geq \inf \sigma_s(L_T-V)$ since $K_T-V$ equals the restriction of $L_T-V$ to zero total momentum.
In total, we have the chain of inequalities
\[
\inf \sigma_{s} (K_{T}-V)\geq \inf \sigma_{s} (L_{T}-V)\geq \inf \sigma_s (D_{T}-V) \geq \inf \sigma (K_T-V).
\]
By assumption, at temperature $T_c$ the left and the right hand side equal zero, thus all of the terms equal zero.
Due to strict monotonicity of $D_T$ in $T$, $T_u=T_c$.
Furthermore, $T_l=T_c$ provided that for $T<T_c$  we have $\inf \sigma_{s} (L_{T}-V)<0$.
This follows from strict monotonicity of $K_T$ in $T$.
\end{proof}

\section{Weak coupling asymptotics}\label{sec:wc}
In this section we introduce the notation needed to study the weak coupling asymptotics.
Furthermore, we prove some results in dimension one, which will be useful for proving both Theorem~\ref{thm:tu-weak-coupling} and \ref{thm:rel-tu-difference}.

To study the weak coupling asymptotics, we apply the Birman-Schwinger principle, similarly to \cite{hainzl_bardeencooperschrieffer_2016, henheik_universality_2023,roos_bcs_2023}.
Recall the definitions of $U$, $N_T$, and $B_T$ from \eqref{def:U}, \eqref{NT}, and \eqref{BT}, respectively.
For $q\in \BR^d$ we shall denote by $N_T(\cdot ,q),B_T(\cdot, q)$ the operators acting on $L^2(\BR^d)$ through multiplication by $N_T(p,q),B_T(p,q)$ in momentum space.
Consider the operators
\begin{equation}
V^{1/2} N_T (\cdot , q) |V|^{1/2} \quad \mathrm{and} \quad V^{1/2} B_T (\cdot , q)  |V|^{1/2}
\end{equation}
on $L^2(\BR^d)$, where $V^{1/2}(r)= \sgn (V(r)) |V(r)|^{1/2}$.
They are precisely the Birman-Schwinger operators corresponding to $U (D_T-\lambda V) U^\dagger$ and $U (L_T-\lambda V) U^\dagger$ at total momentum $q$, respectively.
With the convention that $\sgn (0)=0$, the spectra of the Birman-Schwinger operators satisfy
\begin{equation}\label{eq:BS}
\sgn \inf \sigma  (D_T-\lambda V)=-\sgn \sup_q \sup \sigma  (V^{1/2} N_T (\cdot , q) |V|^{1/2} - \lambda^{-1}\BI),
\end{equation}
where $\BI$ denotes the identity operator.
Similarly for $L_T$ and $B_T$.
We can also restrict to symmetric states on both sides.

Since the Birman-Schwinger operators are bounded for $T$ away from zero, the temperatures $T_u$ and $T_l$ must go to zero in the weak coupling limit.
Therefore, we need to understand the behavior of the Birman-Schwinger operators at low temperatures.
For the operators $V^{1/2} K_T^{-1}  |V|^{1/2}$ and $V^{1/2} B_T (\cdot , q)  |V|^{1/2}$ the low temperature asymptotics are discussed in \cite{hainzl_bardeencooperschrieffer_2016, henheik_universality_2023} and \cite[Proof of Lemma 6.3]{roos_bcs_2023}, respectively.
In both cases, the Birman-Schwinger operators have a singular part, which grows logarithmically in $T$.
The singular part is supported on the Fermi sphere $p^2=\mu$.
Additionally, for $V^{1/2} B_T (\cdot , q)  |V|^{1/2}$ to diverge, the total momentum $q$ must vanish with $T$.

In this paper, we consider the low temperature asymptotics of $V^{1/2} N_T (\cdot , q) |V|^{1/2}$.
We find that in dimensions $2,3$, the operator behaves similarly to $V^{1/2} B_T (\cdot , q)  |V|^{1/2}$, see Section~\ref{sec:pf3}.
In dimension $1$ however, there is an additional divergent contribution from relative momentum $p=0$ and total momentum $q=\mu$.
We shall make this statement precise in the following.

Let $\FT_\mu:L^1(\BR)\to L^2(\{-1,1\})$ be given by $\FT_\mu \psi (p)=(2\pi)^{-1/2} \int_\BR e^{-ip  \sqrt{\mu} x} \psi(x) \dd x$ and let $\FT_0:L^1(\BR)\to \BC$ denote $\FT_0 \psi =(2\pi)^{-1/2} \int_\BR \psi(x) \dd x$.
Define the functions
\begin{equation}\label{eq:def-mn}
m_T(q):=\int_{0}^{\sqrt{3\mu}} B_T(p,q) \dd p \quad \mathrm{and} \quad n_T(q):=\int_{0}^{\sqrt{3\mu}} N_T(p,q) \dd p.
\end{equation}
Note that $m_T(q)\leq n_T(q)$.
Define the families of operators $Q_T(q),W_T(q):L^1(\BR)\to L^\infty(\BR)$ for $q\in \BR$ through
\begin{equation}
Q_T(q) = m_T(q) \FT_\mu^\dagger \FT_\mu \chi_{\max \{T/\mu, |q|/\sqrt{\mu}\}<1/2}
\end{equation}
and
\begin{equation}
W_T(q) = n_T(q) \Big(\FT_\mu^\dagger \FT_\mu \chi_{\max \{T/\mu, |q|/\sqrt{\mu}\}<1/2} +2 \FT_0^\dagger \FT_0 \chi_{\max \{T/\mu, ||q|-\sqrt{\mu}|/\sqrt{\mu}\}<1/2} \Big),
\end{equation}
where $\chi$ denotes the characteristic function.
The operators $Q_T(q)$ are nonzero when $q$ is smaller than $\sqrt{\mu}/2$, and they equal $\FT_\mu^\dagger \FT_\mu$ up to a factor.
The operators $W_T(q)$ have an additional term $\FT_0^\dagger \FT_0$ when $q$ is close to the Fermi sphere.
The following Lemma states, that the operators $Q_T(q)$ and $W_T(q)$ capture the divergence of $B_T(\cdot, q)$ and $ N_T(\cdot, q)$, respectively.
\begin{lemma}\label{lea:1d-approx}
Let  $\mu>0$.
Let $V$ satisfy $(1+ |\cdot|^2) V\in L^1(\BR)$, $V(r)=V(-r)$.
Then
\begin{equation}\label{eq:bt-approx}
\sup_{T>0} \sup_{q \in \BR} \lVert V^{1/2} B_T(\cdot, q) |V|^{1/2} - V^{1/2} Q_T(q) |V|^{1/2} \rVert <\infty
\end{equation}
and
\begin{equation}\label{eq:nt-approx}
\sup_{T>0} \sup_{q \in \BR} \lVert V^{1/2} N_T(\cdot, q) |V|^{1/2} - V^{1/2} W_T(q) |V|^{1/2} \rVert <\infty.
\end{equation}
\end{lemma}
For the operator $B_T$ an analogous result in dimensions 2 and 3 has been proven in \cite[Proof of Lemma 6.3]{roos_bcs_2023}.
The asymptotics of $Q_T(q)$ and $W_T(q)$ are determined by the behaviour of $m_T(q)$ and $n_T(q)$.
The following Lemma states that $m_T(q)$ diverges logarithmically for $q=0$, whereas $n_T(q)$ diverges for $|q|\in \{0,\sqrt{\mu}\}$.
\begin{lemma}\label{lea:1d-w-asy}
Let  $\mu>0$.
Then there is a constant $C$ independent of $T$ and $q$, such that
\begin{equation}\label{eq:mt-up-bd}
m_T(q)  \leq  \mu^{-1/2}\ln\left(\frac{1}{T/\mu + |q|/\sqrt{\mu}}\right)\chi_{\max\{T/\mu,|q|/\sqrt{\mu}\}\leq1/2} +C
\end{equation}
and
\begin{multline}\label{eq:nt-up-bd}
n_T(q)  \leq  \mu^{-1/2}\ln\left(\frac{1}{T/\mu + |q|/\sqrt{\mu}}\right)\chi_{\max\{T/\mu,|q|/\sqrt{\mu}\}\leq1/2} \\
+  \frac{\mu^{-1/2}}{2}\ln\left(\frac{1}{T/\mu +| |q|-\sqrt{\mu}|/\sqrt{\mu}}\right)\chi_{\max\{T/\mu,| |q|-\sqrt{\mu}|/\sqrt{\mu}\}\leq1/2} +C.
\end{multline}
Furthermore, the following lower bounds hold as $T\to 0$:
\[
m_T(0)=n_T(0)\geq  \mu^{-1/2}\ln(\mu/T)+O(1)\quad \textrm{and}\quad n_T(\sqrt{\mu})\geq \frac{1}{2}\mu^{-1/2} \ln(\mu/T)+O(1).
\]
\end{lemma}

In the proofs, we will frequently apply the following bounds.
Using that $N_T$ is monotone decreasing in $T$, observe that
\begin{equation}\label{eq:Mdef}
N_T(p,q)\leq \frac{2}{|(p-q)^2-\mu|+|(p+q)^2-\mu|}=:M(p,q)
\end{equation}
Furthermore, the triangle inequality implies
\begin{equation}
M(p,q)\leq \frac{1}{|p^2+q^2-\mu|}
\end{equation}
Together with the bound $N_T(p,q)\leq 1/2T$ it follows that for all $\epsilon>0$ and $\mu$ there is a constant $C(\mu,\epsilon)>0$ such that for all $T,p,q$
\begin{equation}\label{nt-bound-2}
N_T(p,q) \leq \frac{\chi_{p^2+q^2< \mu+\epsilon}}{2T} + \frac{C(\mu,\epsilon)}{1+p^2+q^2}\chi_{p^2+q^2\geq \mu+\epsilon}.
\end{equation}

\subsection{Proof of Lemma~\ref{lea:1d-approx}}
\begin{proof}[Proof of Lemma~\ref{lea:1d-approx}]
We start by proving \eqref{eq:nt-approx}, i.e.~that the divergence of $N_T$ is captured by $W_T$.
Consider the case that $T\geq \mu/2$ or $|q|\geq \frac{3}{2} \sqrt{\mu}$.
Then $W_T(q)=0$.
The Schwarz inequality implies that for $\psi\in L^2(\BR)$ we have $\lVert  \widehat{V^{1/2}\psi} \rVert_\infty \leq (2\pi)^{-1/2} \lVert  {V^{1/2}\psi} \rVert_1\leq  (2\pi)^{-1/2}\lVert V\rVert_1^{1/2} \lVert \psi \rVert_2$.
Therefore,
\begin{equation}\label{eq:pf-lea31-2}
\lVert V^{1/2} N_T(\cdot, q) |V|^{1/2} \rVert \leq  (2\pi)^{-1}\lVert V\rVert_1 \int_\BR N_T(p,q) \dd p .
\end{equation}
The right hand side is bounded uniformly in $T$ and $q$ since by \eqref{nt-bound-2}, there is a constant $C$ such that $N_T(p,q)\leq C/(1+p^2)$ for all $T\geq \mu/2$ or $|q|\geq \frac{3}{2} \sqrt{\mu}$.

For the approximation of $N_T$ it remains to consider the case $T< \mu/2$ and $|q|< \frac{3}{2} \sqrt{\mu}$.
We prove the following Lemma below.
\begin{lemma}\label{lea:1d-approx-1}
Let $\mu>0$ and $M$ as in \eqref{eq:Mdef}. Then,
\begin{equation}\label{eq:1d-approx-a}
\sup_{|q|\leq\sqrt{\mu}/2 } \int_{-\sqrt{3\mu}}^{\sqrt{3 \mu}} M(p,q) ||p|-\sqrt{\mu}| \dd p <\infty
\end{equation}
and
\begin{equation}\label{eq:1d-approx-b}
\sup_{|q-\sqrt{\mu}|\leq \sqrt{\mu}/2 } \int_{-\sqrt{3\mu}}^{\sqrt{3 \mu}} M(p,q) |p| \dd p <\infty
\end{equation}
\end{lemma}

Let $X_T:=N_T(\cdot,q)-W_T(q)$.
For $T< \mu/2$ and $|q|< \frac{3}{2} \sqrt{\mu}$, the integral kernel of $X_T$ is given by
\begin{multline}
X_T(x,y)=\frac{1}{2\pi} \int_{-\sqrt{3\mu}}^{\sqrt{3 \mu}} N_T(p,q)\Big[ (e^{i p(x-y)}-e^{i \sqrt{\mu} \frac{p}{|p|} (x-y)})\chi_{|q|< \frac{1}{2} \sqrt{\mu}}+  (e^{i p(x-y)}-1)\chi_{||q|-\sqrt{\mu}|< \frac{1}{2} \sqrt{\mu}}\Big] \dd p\\
 + \frac{1}{2\pi} \int_{|p|>\sqrt{3 \mu}} N_T(p,q) e^{i p(x-y)}\dd p
\end{multline}
The second integral is uniformly bounded in $T$ and $q$ according to \eqref{nt-bound-2}.
Using $|e^{i p x}-e^{i q x}|\leq |x||p-q| $ and then Lemma~\ref{lea:1d-approx-1}, we bound the first term by
\[
\frac{|x-y|}{2\pi} \int_{-\sqrt{3\mu}}^{\sqrt{3 \mu}} M(p,q)  \Big[ ||p|-\sqrt{\mu}|\chi_{|q|< \frac{1}{2} \sqrt{\mu}}+  |p|\chi_{||q|-\sqrt{\mu}|< \frac{1}{2} \sqrt{\mu}}\Big]\dd p  < C |x-y|
\]
for some constant $C$ independent of $x,y$.
In particular, the Hilbert-Schmidt norm $\lVert V^{1/2} X_T |V|^{1/2}\rVert_{\textrm{HS}}^2\leq C \int_{\BR^2} |V(x)|(1+|x-y|)^2 |V(y)|\dd x \dd y $.
The latter is finite since we assume that $V, |\cdot|^2V \in L^1$.

In total, we have proved the approximation for $V^{1/2}N_T|V|^{1/2}$ stated in \eqref{eq:nt-approx}.
To prove \eqref{eq:bt-approx}, the approximation for $V^{1/2}B_T|V|^{1/2}$, we proceed similarly.
We start with the case $T\geq \mu/2$ or $|q|\geq \sqrt{\mu}/2$.
Analogously to \eqref{eq:pf-lea31-2} we have
\[
\lVert V^{1/2} B_T(\cdot, q) |V|^{1/2} \rVert \leq  (2\pi)^{-1}\lVert V\rVert_1 \int_\BR B_T(p,q) \dd p .
\]
For $T\geq \mu/2$ or $|q|\geq \frac{3}{2}\sqrt{\mu}$ boundedness follows from \eqref{nt-bound-2}, as for $N_T$.
In contrast to $N_T$, the integral of $B_T$ is also bounded for $\sqrt{\mu}/2\leq |q|\leq \frac{3}{2}\sqrt{\mu}$, i.e.
\[
\sup_{T>0,|q|\geq \sqrt{\mu}/2} \int_\BR B_T(p,q) \dd p<\infty,
\]
which was shown in the proof of \cite[Lemma 4.4]{hainzl_boundary_2022}.

For the case $T<\mu/2$ and $q<\sqrt{\mu}/2$, we define $X_T=B_T(\cdot, q)-Q_T(q)$.
The integral kernel is
\begin{equation}
X_T(x,y)=\frac{1}{2\pi} \int_{-\sqrt{3\mu}}^{\sqrt{3 \mu}} B_T(p,q) (e^{i p(x-y)}-e^{i \sqrt{\mu} \frac{p}{|p|} (x-y)}) \dd p\\
 + \frac{1}{2\pi} \int_{|p|>\sqrt{3 \mu}} B_T(p,q) e^{i p(x-y)}\dd p
\end{equation}
Again, the second integral is uniformly bounded in $T$ and $q$ by \eqref{nt-bound-2}, while it again follows from Lemma~\ref{lea:1d-approx-1} that the first is bounded by $C |x-y|$ for a constant $C$ independent of $T$ and $q$.
With the same argument as before, one concludes that the Hilbert-Schmidt norm of $X_T$ is bounded uniformly in $T$ and $q$.
\end{proof}

\begin{proof}[Proof of Lemma~\ref{lea:1d-approx-1}]
The intuition is that $M$ diverges linearly at $p=0$ or $p=\sqrt{\mu}$ if $q=\sqrt{\mu}$ or $q=0$, respectively.
We need to show that the additional factor of $||p|-\sqrt{\mu}|$ or $|p|$ is sufficient to cancel the divergence of the integral uniformly in $q$.
For fixed $q$, we will distinguish cases depending on the signs of $(p-q)^2-\mu$ and $(p+q)^2-\mu$.

We start by proving \eqref{eq:1d-approx-a}.
Since $M$ is an even function in both $p$ and $q$, it suffices to consider $p,q\geq 0$.
We split the following integral into three parts
\begin{multline}
\int_{0}^{\sqrt{3 \mu}} M(p,q) ||p|-\sqrt{\mu}| \dd p = \int_{0}^{\sqrt{\mu}-q} \frac{\sqrt{\mu}-p}{\mu-p^2-q^2} \dd p
+\int_{\sqrt{\mu}-q}^{\sqrt{\mu}+q} \frac{|p-\sqrt{\mu}|}{2pq} \dd p\\
+\int_{\sqrt{\mu}+q}^{\sqrt{3 \mu}} \frac{p-\sqrt{\mu}}{p^2+q^2-\mu}  \dd p,
\end{multline}
where in the first part both $(p-q)^2-\mu$ and $(p+q)^2-\mu$ are negative, in the second part $(p+q)^2-\mu$ becomes positive and in the third part both are positive.
The last term is bounded by
\[
\int_{\sqrt{\mu}}^{\sqrt{3 \mu}} \frac{p-\sqrt{\mu}}{p^2-\mu}  \dd p<\infty.
\]
For the middle term, we use that $(p-\sqrt{\mu})/p \leq 1$ due to the constraints $\frac{\sqrt{\mu}}{2}\leq \sqrt{\mu}-q\leq p \leq \sqrt{\mu}+q \leq 3\frac{\sqrt{\mu}}{2}$.
The factor $1/2q$ cancels with the size of the integration domain and the integral is thus bounded by 1.
For the first term we carry out the integration and obtain
\[
\int_{0}^{\sqrt{\mu}-q} \frac{\sqrt{\mu}-p}{\mu-p^2-q^2} \dd p= \frac{\sqrt{\mu}}{2} \frac{\ln(\sqrt{\mu}+\sqrt{\mu-q^2})}{\sqrt{\mu-q^2}}+\Big(\frac{1}{2}-\frac{\sqrt{\mu}}{2 \sqrt{\mu-q^2}} \Big)\ln (q) -\frac{1}{2} \ln\left(\frac{\sqrt{\mu}+q}{2}\right)
\]
Clearly, the first and the last summand are bounded uniformly for $0\leq q \leq \frac{\sqrt{\mu}}{2}$.
We need to check that the middle term is bounded as $q\to 0$, but this follows from
\[ 1-\frac{\sqrt{\mu}}{ \sqrt{\mu-q^2}}=\frac{-q^2}{\sqrt{\mu-q^2}(\sqrt{\mu-q^2}+\sqrt{\mu})}.
\]
This completes the proof of \eqref{eq:1d-approx-a}.

Now we shall prove \eqref{eq:1d-approx-b}.
Again since $M$ is even, it suffices to consider $p\geq 0$ and we split
\begin{equation}
\int_{0}^{\sqrt{3 \mu}} M(p,q) p \dd p = \int_{0}^{|\sqrt{\mu}-q|} \frac{p}{|\mu-p^2-q^2|} \dd p
+\int_{|\sqrt{\mu}-q|}^{\sqrt{\mu}+q} \frac{1}{2q} \dd p
+\int_{\sqrt{\mu}+q}^{\sqrt{3 \mu}} \frac{p}{p^2+q^2-\mu}  \dd p.
\end{equation}
The first term equals $\frac{1}{2}\Big| \ln\Big(\frac{\sqrt{\mu}+q}{2q}\Big)\Big|$.
The second term is bounded above by 1.
The third term equals $\frac{1}{2}\ln\Big(\frac{2\mu+q^2}{2q(\sqrt{\mu}+q)}\Big)$.
Hence, all of them are uniformly bounded for $\sqrt{\mu}/2\leq q\leq 3\sqrt{\mu}/2$.
\end{proof}

\subsection{Proof of Lemma~\ref{lea:1d-w-asy}}
\begin{proof}[Proof of Lemma~\ref{lea:1d-w-asy}]
For $T\to 0$, the function $N_T$ converges pointwise to $M$.
For $q\in\{0,\sqrt{\mu}\}$, the integral $\int M(p,q)\dd p$ diverges.
Hence, $n_T$ will diverge for $T$ small and $q$ close to zero or $\sqrt{\mu}$.
For $m_T$, there is no divergence at $q=\sqrt{\mu}$, since $B_T$ converges to zero instead of $M$ when $(p+q)^2-\mu$ and $(p-q)^2-\mu$ have opposite signs \cite{hainzl_boundary_2022}.

To compute the upper bound on $n_T$, the main idea is to find the origin of the divergence.
Everywhere else the integral of $N_T$ can be bounded uniformly in $T$ and $q$.
For the diverging parts we bound the function $N_T$ above by a function for which the integral can be computed explicitly.

For the lower bounds, the idea is to reduce them to similar lower bounds which are already known \cite{hainzl_critical_2008,hainzl_boundary_2022}.

\textbf{Upper bound.}
We shall first prove the upper bound on $n_T(q)$.
Since $n_T(-q)=n_T(q)$ it suffices to consider $q>0$.
We will use the bound $N_T(p,q)\leq \min\{1/2T, M(p,q)\}$ and $M(p,q)\leq 1/|p^2+q^2-\mu|$.
It follows that for $T\geq \mu/2$ or $q\geq 3 \sqrt{\mu}/2$, the function $N_T(p,q)$ is bounded above by $1/\mu$.
Thus $n_T(q)\leq C_1$ for a constant $C_1$ and it suffices to restrict to $T\leq\mu/2$ and $q\leq 3 \sqrt{\mu}/2$ from now on.

For $\sqrt{\mu}/2 \leq q\leq 3 \sqrt{\mu}/2$, we estimate as follows
\begin{equation}\label{eq:nt-bound-1}
n_T(q)\leq \int_0^{|q-\sqrt{\mu}|}M(p,q) \dd p +\int_{|q-\sqrt{\mu}|}^{|q-\sqrt{\mu}|+T/\sqrt{\mu}}  \frac{1}{2T} \dd p +  \int_{|q-\sqrt{\mu}|+T/\sqrt{\mu}}^{\sqrt{\mu}} M(p,q) \dd p+\int_{\sqrt{\mu}}^{\sqrt{3\mu}} M(p,q) \dd p
\end{equation}
The second term is clearly bounded by a constant independent of $T$ and $q$.
The same is true for the last term, since there $p^2+q^2\geq \mu+\mu/4$, and thus $M(p,q)\leq 4/\mu$.
For the first term, carrying out the integration gives
\[
\begin{matrix}
\frac{1}{\sqrt{\mu-q^2}} \artanh\left(\frac{(\sqrt{\mu}-q)^{1/2}}{(\sqrt{\mu}+q)^{1/2}}\right) & \textrm{if}\  \sqrt{\mu}/2\leq q< \sqrt{\mu},\\
0 & \textrm{if}\   |q|= \sqrt{\mu},\\
\frac{1}{\sqrt{q^2-\mu}} \arctan\left(\frac{(q-\sqrt{\mu})^{1/2}}{(\sqrt{\mu}+q)^{1/2}}\right)& \textrm{if}\ \sqrt{\mu} < q \leq 3\sqrt{\mu}/2.
\end{matrix}
\]
This is bounded by
\[
\begin{matrix}
\frac{1}{\sqrt{\mu}+q} \sup_{0< x\leq 1/\sqrt{3}} x^{-1}\artanh x & \textrm{if}\ \sqrt{\mu}/2\leq q< \sqrt{\mu},\\
\frac{1}{\sqrt{\mu}+q} \sup_{x>0 } x^{-1}\arctan x & \textrm{if}\ \sqrt{\mu} < q \leq 3\sqrt{\mu}/2.
\end{matrix}
\]
Both of the suprema give finite values.
Therefore, the first term in \eqref{eq:nt-bound-1} is bounded by a constant independent of $T,q$.
The third term in \eqref{eq:nt-bound-1} equals
\[
 \int_{|q-\sqrt{\mu}|+T/\sqrt{\mu}}^{\sqrt{\mu}} \frac{1}{2pq} \dd p=\frac{1}{2q}\ln \left(\frac{\sqrt{\mu}}{|q-\sqrt{\mu}|+T/\sqrt{\mu}}\right)
\]
For $|q-\sqrt{\mu}|, T/\sqrt{\mu}\leq \sqrt{\mu}/2$ observe that
\begin{equation}\label{eq:nt-bound-4}
\left|\left(\frac{1}{2\sqrt{\mu}}-\frac{1}{2q}\right)\ln \left(\frac{\sqrt{\mu}}{|q-\sqrt{\mu}|+T/\sqrt{\mu}}\right)\right|
 \leq  \frac{|q-\sqrt{\mu}|}{2 q \sqrt{\mu}} \ln \left(\frac{\sqrt{\mu}}{|q-\sqrt{\mu}|}\right)
\leq \frac{1}{\sqrt{\mu}} \sup_{x\geq 2} \frac{\ln x}{x}<\infty.
\end{equation}
Therefore, \eqref{eq:nt-bound-1} is bounded above by $\frac{1}{2\sqrt{\mu}}  \ln \left(\frac{\sqrt{\mu}}{|q-\sqrt{\mu}|+T/\sqrt{\mu}}\right) +C_2$ for some $C_2<\infty$.

It remains to consider $0\leq q < \sqrt{\mu}/2$.
Similarly to the previous case, we bound
\begin{equation}\label{eq:nt-bound-2}
n_T(q)\leq \int_0^{\sqrt{\mu}-q-T/\sqrt{\mu}} M(p,q) \dd p+  \int_{\sqrt{\mu}-q}^{\sqrt{\mu}+q} M(p,q) \dd p+ \int_{\sqrt{\mu}+q+T/\sqrt{\mu}}^{\sqrt{3\mu}} M(p,q) \dd p+C
\end{equation}
for some $C<\infty$, where $C$ comes from the intervals of order $T$ where we apply the bound $N_T(p,q)\leq 1/2T$.
The first term equals
\begin{equation}\label{eq:nt-bound-3}
 \int_0^{\sqrt{\mu}-q-T/\sqrt{\mu}} \frac{1}{\mu-p^2-q^2}\dd p=\frac{1}{2\sqrt{\mu-q^2}} \ln\left(1+2\frac{\sqrt{\mu}-q-T/\sqrt{\mu}}{\sqrt{\mu-q^2}-\sqrt{\mu}+q+T/\sqrt{\mu}}\right)
\end{equation}
Observe that $\sqrt{\mu-q^2}-\sqrt{\mu}\geq -q^2/\sqrt{\mu} $, which for $q\leq \sqrt{\mu}/2$ is bounded below by $-q/2$.
Due to the constraints $q/\sqrt{\mu},T/\mu \leq 1/2$, we have $1\leq \frac{1}{q/\sqrt{\mu}+T/\mu}$.
Combining these bounds, the expression in \eqref{eq:nt-bound-3} is bounded above by
\[
\frac{1}{2\sqrt{\mu-q^2}} \ln\left(\frac{5}{q/\sqrt{\mu}+T/\mu}\right)\leq \frac{1}{2\sqrt{\mu}} \ln\left(\frac{5}{q/\sqrt{\mu}+T/\mu}\right)+C
\]
for some finite constant $C$ independent of $q$ and $T$, which follows from an estimate analogous to \eqref{eq:nt-bound-4}.
The second term in \eqref{eq:nt-bound-2} equals $0$ for $q=0$ and for $q> 0$
\[
\int_{\sqrt{\mu}-q}^{\sqrt{\mu}+q} \frac{1}{2pq} \dd p=\frac{1}{2q}\ln \left(1+\frac{2q}{\sqrt{\mu}-q}\right) \leq \frac{1}{\sqrt{\mu}-q} \sup_{x>0} x^{-1} \ln(1+ x).
\]
For $0\leq q \leq  \sqrt{\mu}/2$ this is bounded by a finite constant independent of $q$.
The third term in \eqref{eq:nt-bound-2} equals
\[
 \int_{\sqrt{\mu}+q+T/\sqrt{\mu}}^{\sqrt{3\mu}} \frac{1}{p^2+q^2-\mu} \dd p\leq  \int_{\sqrt{\mu}+q+T/\sqrt{\mu}}^{\sqrt{3\mu}} \frac{1}{2\sqrt{\mu}(p-\sqrt{\mu})} \dd p=\frac{1}{2\sqrt{\mu}}\ln \left(\frac{\sqrt{3}-1}{q/\sqrt{\mu}+T/\mu}\right).
\]
Therefore, the expression in \eqref{eq:nt-bound-2} is bounded by
\[
\frac{1}{\sqrt{\mu}}\ln \left(\frac{1}{q/\sqrt{\mu}+T/\mu}\right) + C_3
\]
for some finite constant $C_3$.
Collecting all the bounds for $n_T(q)$, we obtain \eqref{eq:nt-up-bd} with $C=\max\{ C_1,C_2,C_3\}$.

For the upper bound on $m_T(q)$, since $m_T(q)\leq n_T(q)$ we apply the bound \eqref{eq:nt-up-bd} on $n_T(q)$ derived above for $|q|\leq \sqrt{\mu}/2$.
For $|q|\geq \sqrt{\mu}/2$, it was shown in \cite[Lemma 4.4]{hainzl_boundary_2022} that $ \sup_{T>0,|q|\geq \sqrt{\mu}/2} m_T(q) = \sup_{T>0,|q|\geq \sqrt{\mu}/2} \int_\BR B_T(p,q) \dd p<\infty$.
Hence, the upper bound \eqref{eq:mt-up-bd} follows.

\textbf{Lower bounds.}
The desired lower bound for
\[
m_T(0)=n_T(0)= \int_{0}^{\sqrt{3\mu}} \frac{\tanh((p^2-\mu)/2T}{p^2-\mu} \dd p
\]
follows from \cite[Lemma 3.5]{hainzl_boundary_2022}.
To prove the lower bound for $n_T(\sqrt{\mu})$, we first observe that if $|x|\geq |y|$, then due to monotonicity of $x/\tanh (x)$
\[
\frac{2}{\frac{x}{\tanh(x)}+\frac{y}{\tanh(y)}}\geq \frac{\tanh(x)}{x}.
\]
This implies that
\[
n_T(\sqrt{\mu})\geq \int_0^{\sqrt{\mu}/2} \frac{\tanh\Big(\frac{p(2\sqrt{\mu}+p)}{2T}\Big)}{p(2\sqrt{\mu}+p)} \dd p=\int_0^{5\mu/4} t^{-1} \tanh\Big(\frac{t}{2T}\Big) \frac{1}{2\sqrt{\mu+t}}\dd t,
\]
where we substituted $t=p(2\sqrt{\mu}+p)$.
We can rewrite the latter as
\[
 \frac{1}{2\sqrt{\mu}}\int_0^{5\mu/4} t^{-1}\tanh\Big(\frac{t}{2T}\Big)\dd t - \frac{1}{2} \int_0^{5\mu/4} \frac{\tanh\Big(\frac{t}{2T}\Big)}{\sqrt{\mu+t}\sqrt{\mu}(\sqrt{\mu+t}+\sqrt{\mu})} \dd t
\]
The second term is bounded as $T\to 0$ and the first term asymptotically equals $\frac{1}{2\sqrt{\mu}} \ln \mu/T +O(1)$ as was shown in \cite[Lemma 1]{hainzl_critical_2008}.
In total, we have $n_T(\sqrt{\mu})\geq \frac{1}{2\sqrt{\mu}} \ln \mu/T +O(1)$.
\end{proof}

\section{Proof of Theorem~\ref{thm:tu-weak-coupling}}\label{sec:pf2}
\begin{proof}[Proof of Theorem~\ref{thm:tu-weak-coupling}]
The goal is to compute the weak coupling asymptotics of $T_l$ and $T_u$ in dimension one.
Recall the operators $\FT_\mu$ and $\FT_0$ defined above \eqref{eq:def-mn}.
For $x\in\{0,\mu\}$, the operator $V^{1/2}\FT_x^\dagger \FT_x |V|^{1/2}$ has the same spectrum as $\FT_x V \FT_x^\dagger$ up to zero.
For $x=0$, the spectrum is just the number $\frac{1}{2}e_0^s=(2\pi)^{-1/2}\widehat{V}(0)$.
The operator $\FT_\mu V \FT_\mu^\dagger$ is precisely $\mathcal{V}_\mu$ defined in \eqref{eq:Vmu}.
In particular, $\sup \sigma_s (\FT_\mu V \FT_\mu^\dagger)=e_\mu^s = \frac{\widehat{V}(0)+\widehat{V}(2\sqrt{\mu})}{(2\pi)^{1/2}}$.
By assumption, we have $e_\mu^s>0$.
In total, we observed that $\sup \sigma_s(V^{1/2}\FT_\mu^\dagger \FT_\mu |V|^{1/2})=e_\mu^s>0$ and $\sup \sigma_s(V^{1/2}\FT_0^\dagger \FT_0 |V|^{1/2})=\frac{1}{2}e_0^s$.

Combining Lemmas~\ref{lea:1d-approx} and \ref{lea:1d-w-asy}, we obtain that for $T\to 0$
\begin{multline}
\sup_q \sup \sigma_s(V^{1/2} N_T(\cdot,q) |V|^{1/2})=\sup_q \sup \sigma_s(V^{1/2} W_T(q) |V|^{1/2})+O(1)\\
=\mu^{-1/2} \ln(\mu/T) \max\left\{e_\mu^s, \frac{1}{2}e_0^s\right\} +O(1)
\end{multline}
and similarly
\[
\sup_q \sup \sigma_s(V^{1/2} B_T(\cdot,q) |V|^{1/2})=\mu^{-1/2} \ln(\mu/T) e_\mu^s +O(1).
\]
According to the Birman-Schwinger principle \eqref{eq:BS}, $\inf \sigma_s(D_{T_u(\lambda)}-\lambda V)=0$ is equivalent to $\sup_q \sup \sigma_s(V^{1/2} N_{T_u(\lambda)}(\cdot,q) |V|^{1/2})=\lambda^{-1}$.
In particular, for all $\lambda>0$ the temperature $T_u(\lambda)$ is positive.
At weak coupling we obtain the asymptotics
\[ T_u(\lambda)=\mu e^{-\mu^{1/2}/(\lambda \max\{e_\mu^s,\frac{1}{2}e_0^s\}) +O(1)}
\]
as $\lambda\to 0$.
Analogously, we obtain for $\lambda \to 0$
\[
T_l(\lambda)=\mu e^{-\mu^{1/2}/(\lambda e_\mu^s) +O(1)}.
\]
\end{proof}

\section{Proof of Theorem~\ref{thm:rel-tu-difference}}\label{sec:pf3}
The goal is to prove
\[
\lim_{\lambda\to 0}\frac{T_u^1(\lambda)-T_c^0(\lambda)}{T_c^0(\lambda)}=0.
\]
In \cite{roos_bcs_2023} a similar statement was proved for $T_l^1$ instead of $T_u^1$, i.e.
\[
\lim_{\lambda\to 0}\frac{T_l^1(\lambda)-T_c^0(\lambda)}{T_c^0(\lambda)}=0.
\]
For a large part of the proof, one can follow the same strategy and just replace $T_l$ by $T_u$ and $B_T$ by $N_T$.
Therefore, instead of repeating the whole proof, we only mention which changes need to be made to \cite[Section 6]{roos_bcs_2023}.
Recall that $V\geq 0$ by assumption.
This assumption is necessary for the proof strategy in \cite{roos_bcs_2023} to work.
The function
\begin{equation}
E_T(q)=\sup_{q\in \BR^d}\lVert V^{1/2} N_T(\cdot , q) V^{1/2} \rVert_s -\lVert V^{1/2} N_T(\cdot , q) V^{1/2} \rVert_s
\end{equation}
plays an important role.
Here, $\lVert \cdot \rVert_s$ denotes the operator norm on the space of even $L^2$-functions.
The argument in \cite{roos_bcs_2023} is based on the observation that $E_T(q)\geq 0$ as well as three key Lemmas, \cite[Lemma 6.1, Lemma 6.2, Lemma 6.3]{roos_bcs_2023}.

The proofs of \cite[Lemma 6.1, Lemma 6.2]{roos_bcs_2023} rely on the bound $B_T\leq M$ only, where $M$ was defined in \eqref{eq:Mdef}.
Since the function $N_T$ is also bounded above by $M$, the proofs go through just replacing $B_T$ by $N_T$.

For \cite[Lemma 6.3]{roos_bcs_2023}, however, while the statement is still true after replacing $B_T$ with $N_T$, some non-trivial changes need to be made to the proof.
The following Lemma is left to prove.
\begin{lemma}
Let $\mu>0$, $d\in\{1,2,3\}$. 
Let $V\geq 0$ satisfy Assumptions~\ref{as1}, $(1+|\cdot|^2)V \in L^1(\BR)$ for $d=1$, $V\in L^1(\BR^d)$ for $d\in\{2,3\}$,  and $e_\mu^s>\max\{0,e_\mu^a\}$.
Let $0<\epsilon< \sqrt{\mu}$. There are constants $c_1,c_2,T_1>0$ such that for $0<T<T_1$ and $|q|>\epsilon$ we have $E_T(q)>c_1|\ln(c_2/T)|$.
\end{lemma}

\begin{proof}
First we need to understand the asymptotics of $\sup_{q\in \BR}\lVert V^{1/2} N_T(\cdot , q) V^{1/2} \rVert_s$ for $T\to 0$.
We shall argue that at small enough temperatures, the supremum is attained at $q=0$.
Recall the assumption $e_\mu^s>\max\{0,e_\mu^a\}$.
According to Remark~\ref{rem:uniquetc}, there is a $\lambda_0$ such that for $0<\lambda\leq \lambda_0$ it holds that $T_c^0(\lambda)=T_u^0(\lambda)$ and the ground state of $K_{T_c^0(\lambda)}-\lambda V$ is even.
Let $T_0=T_c^0(\lambda_0)$.
Due to the monotonicity of $N_T$ in $T$, for all $T<T_0$ there is a $\lambda<\lambda_0$ such that $T=T_c^0(\lambda)=T_u^0(\lambda)$.
The Birman-Schwinger principle for $D_T-\lambda V$ and $K_T-\lambda V$ implies that
\[
\sup_q \lVert V^{1/2} N_T(\cdot , q) V^{1/2} \rVert_s = \frac{1}{\lambda}=\lVert V^{1/2} N_T(\cdot , 0) V^{1/2} \rVert_s.
\]
Hence, the supremum is attained at momentum $q=0$ for $T<T_0$.
For $q=0$ it was computed in \cite{hainzl_bardeencooperschrieffer_2016, henheik_universality_2023} that for $T\to 0$
\[
\lVert V^{1/2} N_T(\cdot , 0) V^{1/2} \rVert_s=e_\mu^s \mu^{d/2-1} \ln \left(\frac{\mu}{T}\right)+O(1).
\]
We need to show for $|q|>\epsilon$ that the second term in $E(q)$, $\lVert V^{1/2} N_T(\cdot , q) V^{1/2} \rVert_s$ grows more slowly.

\textbf{Dimension one:}
In dimension one, apart from the singularity at total momentum zero, there is another singularity at $|q|=\sqrt{\mu}$.
We need to check that the divergence at the latter singularity is slower than at $q=0$.
Combining Lemmas~\ref{lea:1d-approx} and \ref{lea:1d-w-asy}, if $|q|>\epsilon$ we obtain that
\[
\sup_{T>0} \sup_{|q|>\epsilon} \lVert V^{1/2} N_T(\cdot, q) V^{1/2}\rVert_s  \leq  \frac{1}{2}e_0^s \mu^{-1/2} \ln (\mu/T) +O(1),
\]
where $ e_0^s =2 (2\pi)^{-1/2} \widehat V (0)>0$ since $V\geq 0$.
The assumption $e_\mu^s>e_\mu^a$ implies $\widehat V(2\sqrt{\mu})>0$, and thus $e_\mu^s> \frac{1}{2} e_0^s$ and the claim follows.

\textbf{Dimensions two and three:}
We show that in dimensions two and three there is no divergence away from zero total momentum, i.e.~$\sup_{T>0} \sup_{|q|>\epsilon} \lVert V^{1/2} N_T(\cdot, q) V^{1/2}\rVert  <\infty$.
This follows if we prove that $ \sup_{|q|>\epsilon} \lVert V^{1/2} M(\cdot, q) V^{1/2}\rVert  <\infty$, where $M$ was defined in \eqref{eq:Mdef}.

\begin{figure}
\centering
\begin{tikzpicture}[scale=0.5]
\tikzmath{\smu=5;\l=10;\q=1.8;\t=35;\d=0.2;}

 \draw[fill, fill opacity=0.3, gray,even odd rule] (-\q,0) circle (\smu) (\q,0) circle (\smu);
\draw[black, name path= c1] (-\q,0) circle (\smu);
\draw[black, name path =c2]  (\q,0) circle (\smu);

\draw (0,2*\l/3) node[left]{$|\tilde p|$};
\draw (\l,0) node[below]{$p_1$};
\draw (\q,0) node[below]{$\vert q \vert$};
\draw (-\q,0) node[below]{$-\vert q \vert$};
\draw (0,\smu) node[above right]{$\sqrt{\mu}$};

\draw (0,\smu/3) node[above right]{$A_3$};
\draw (\smu+\q,\smu/3) node[above right]{$A_3$};
\draw (\smu-\q,\smu/3) node[above right]{$A_2$};
\draw (-\smu+\q,\smu/3) node[above left]{$A_2$};

\draw[black, thick] ($ (\q,0)+(0,\d)$) -- ($ (\q,0)-(0,\d)$);
\draw[black, thick] ($ (-\q,0)+(0,\d)$) -- ($ (-\q,0)-(0,\d)$);
\draw[black, thick] ($ (0,\smu)+(\d,0)$) -- ($ (0,\smu)-(\d,0)$);

 \draw[black, dashed, decoration={markings, mark=at position 1 with {\arrow[scale=2,>=stealth]{>}}},
        postaction={decorate}]
 (-\l,0)--(\l,0);
 \draw[black, dashed, decoration={markings, mark=at position 1 with {\arrow[scale=2,>=stealth]{>}}},
        postaction={decorate}]
 (0,-2*\l/3)--(0,2*\l/3);
\end{tikzpicture}

\caption{Two circles of radius $\sqrt{\mu}$, centered at $(-\vert q \vert,0)$ and $(\vert q \vert,0)$.
Assume that the coordinate system is chosen such that $q=(|q|,0,0)$ and write the vector $p=(p_1, \tilde p)$.
The points in $A_2$ lie in one of the circles, but not in the other. Thus $A_2$ is the shaded area.
The points in $A_3$ either lie in both circles or outside both of them, i.e.~the white area.
The sketch is adapted from \cite{roos_bcs_2023}.}
\label{fig:sketch}

\end{figure}
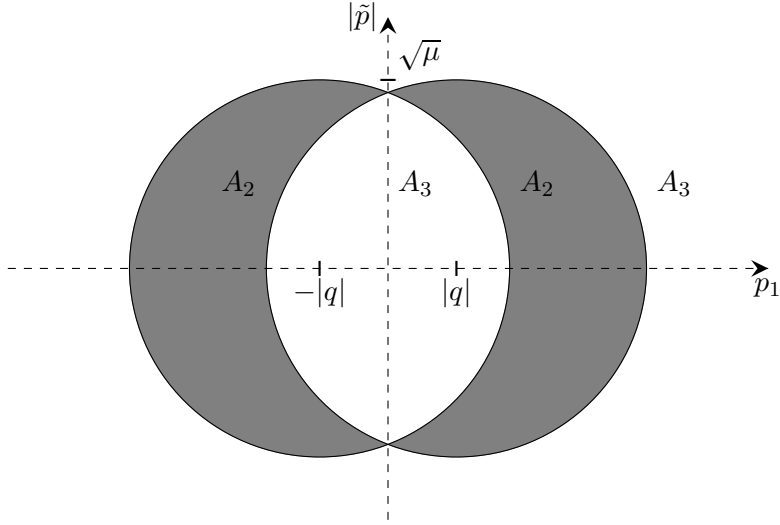

We define the sets $A_1=\{p \in \BR^d | p^2>3\mu\}$, $A_2=\{p \in \BR^d | p^2<3\mu, ((p-q)^2-\mu)((p+q)^2-\mu)<0 \}$, and $A_3=\{p \in \BR^d | p^2<3\mu, ((p-q)^2-\mu)((p+q)^2-\mu)>0 \}$.
The sets $A_2$ and $A_3$ are sketched in Figure~\ref{fig:sketch}.
Let $\chi_{A_j}$ denote the corresponding characteristic functions.
To bound the norm of $V^{1/2} M(\cdot, q) V^{1/2}$, we split the operator into the sum of the three operators $V^{1/2} M(\cdot, q) \chi_{A_j}(\cdot) V^{1/2}$ for $1\leq j\leq 3$ and bound each of the norms.

Recall that $M(\cdot,q)\chi_{A_1}(p)\leq C/(1+p^2)$ according to \eqref{nt-bound-2}.
It follows from the Hardy-Littlewood-Sobolev inequality that $\lVert V^{1/2} \frac{1}{1-\Delta} V^{1/2}\rVert <\infty$, where $\Delta$ is the Laplacian \cite{hainzl_bardeencooperschrieffer_2016,henheik_universality_2023,lieb_analysis_2001}.
Therefore,
\[
\sup_{|q|>\epsilon} \lVert V^{1/2} M(\cdot, q) \chi_{A_1}(\cdot) V^{1/2}\rVert  <\infty.
\]
To bound
\[
\sup_{|q|>\epsilon} \lVert V^{1/2} M(\cdot, q) \chi_{A_2}(\cdot) V^{1/2}\rVert,
\]
one uses the Schwarz inequality to obtain
\begin{equation}\label{eq:23d_pf_2}
\sup_{|q|>\epsilon} \lVert V^{1/2} M(\cdot, q) \chi_{A_2}(\cdot) V^{1/2}\rVert  \leq \lVert V\rVert_1 \sup_{|q|>\epsilon}\int_{\BR^d} M(p,q)\chi_{A_2}(p)\dd p.
\end{equation}
To bound the latter, first notice that we may assume $q=(|q|,0)$ without loss of generality by rotating the coordinate system of $p$ and $q$.
Observe that
\[
 M(p,(|q|,0))\chi_{A_2}(p)=\frac{1}{2 |p_1| |q|}.
\]
It is evident from Figure~\ref{fig:sketch} that in the domain $A_2$, we have $\max\{0,|q|-\sqrt{\mu}\}\leq |p_1|\leq \sqrt{\mu}+|q|$ and
\[
|\tilde p|\in ( \sqrt{\max\{0, \mu-(|p_1|+ |q|)^2\}}, \sqrt{\mu-(|p_1|-|q|)^2}).
\]
We first carry out the integration over the angular part of $\tilde p$.
In $d=3$, the angular part gives a contribution $2 \pi |\tilde p| \leq 2 \pi \sqrt{\mu}=:c_3$, in $d=2$ we get a factor $c_2=2$.
Integrating now over the radial part of $\tilde p$ and using the symmetry in $p_1\to -p_1$, gives
\begin{equation}\label{eq:23d_pf_1}
\int_{\BR^d} M(p,(|q|,0))\chi_{A_2}(p)\dd p \leq \frac{c_d}{|q|}\int_{\max\{0,|q|-\sqrt{\mu}\}}^{\sqrt{\mu}+|q|} \frac{1}{p_1} \Big(\sqrt{\mu-(p_1-|q|)^2}- \sqrt{\max\{0, \mu-(p_1+ |q|)^2\})}\Big)\dd p_1.
\end{equation}
Let us distinguish the cases $q\geq \sqrt{\mu}$ and $q \leq \sqrt{\mu}$.
For $q\geq \sqrt{\mu}$ the right hand side of \eqref{eq:23d_pf_1} reduces to
\[
\frac{c_d}{|q|}\int_{|q|-\sqrt{\mu}}^{\sqrt{\mu}+|q|} \frac{1}{p_1}\sqrt{\mu-(p_1-|q|)^2}\dd p_1.
\]
Using that $\mu-(p_1-|q|)^2\leq 2 p_1 |q|$ for $|q|\geq \sqrt{\mu}$, this is bounded by
\[
\frac{\sqrt{2} c_d}{|q|^{1/2}}\int_{|q|-\sqrt{\mu}}^{\sqrt{\mu}+|q|} \frac{1}{p_1^{1/2}}\dd p_1 \leq  \frac{2\sqrt{2} c_d}{|q|^{1/2}}(\sqrt{\mu}+|q|)^{1/2}\leq 4 c_d.
\]

Now, consider the case $q\leq \sqrt{\mu}$.
We rewrite the right hand side of \eqref{eq:23d_pf_1} as
\begin{multline}\label{eq:23d_pf_3}
\frac{c_d}{|q|}\int_{0}^{\sqrt{\mu}-|q|} \frac{1}{p_1} \Big(\sqrt{\mu-(p_1-|q|)^2}- \sqrt{\mu-(p_1+ |q|)^2}\Big)\dd p_1\\
 + \frac{c_d}{|q|}\int_{\sqrt{\mu}-|q|}^{\sqrt{\mu}+|q|} \frac{1}{p_1} \sqrt{\mu-(p_1-|q|)^2}\dd p_1.
\end{multline}
For the first term, we shall use that
\begin{multline}
\sqrt{\mu-(p_1-|q|)^2}- \sqrt{\mu-(p_1+ |q|)^2}=\frac{4 p_1 |q|}{\sqrt{\mu-(p_1-|q|)^2}+  \sqrt{\mu-(p_1+ |q|)^2}}\\
\leq \frac{4 p_1 |q|}{ \sqrt{\mu-(p_1+ |q|)^2}}
\leq \frac{4 p_1 |q|}{\mu^{1/4}(\sqrt{\mu}-|q|-p_1)^{1/2}}
\end{multline}
For the second term in \eqref{eq:23d_pf_3}, notice that for $p_1\geq \sqrt{\mu}-|q|$, we have
\[
\mu-(p_1-|q|)^2 =\mu-p_1^2-q^2+2 p_1 |q| \leq  \mu-(\sqrt{\mu}-|q|)^2-q^2+2p_1|q| = 2|q|( \sqrt{\mu}-|q|+ p_1)
\]
and $p_1 \geq (p_1+\sqrt{\mu}-|q|)/2$.
In total, we can bound \eqref{eq:23d_pf_3} by
\begin{multline}
c_d\int_{0}^{\sqrt{\mu}-|q|} \frac{4 }{\mu^{1/4}(\sqrt{\mu}-|q|-p_1)^{1/2}}\dd p_1
 + \frac{2^{1/2}c_d}{|q|^{1/2}}\int_{0}^{\sqrt{\mu}+|q|} \frac{2}{(p_1+\sqrt{\mu}-|q|)^{1/2}} \dd p_1\\
\leq \frac{8c_d}{\mu^{1/4}} (\sqrt{\mu}-|q|)^{1/2}+ \frac{8 c_d \mu^{1/4}}{|q|^{1/2}} .
\end{multline}
Since we assumed that $\epsilon \leq |q| \leq \sqrt{\mu}$, this is bounded by $8 c_d(1+\mu^{1/4}/\sqrt{\epsilon})$.
This completes the proof of $\sup_{|q|>\epsilon} \lVert V^{1/2} M(\cdot, q) \chi_{A_2}(\cdot) V^{1/2}\rVert <\infty$.

We are left with showing
\[
\sup_{|q|>\epsilon} \lVert V^{1/2} M(\cdot, q) \chi_{A_3}(\cdot) V^{1/2}\rVert<\infty.
\]
This follows from elementary computations, as for $A_2$.
In fact, the computation for $\sup_{|q|>\sqrt{\mu}/2}$ is spelled out in \cite[Proof of Lemma 6.9 part (iv)]{roos_bcs_2023}.
The method carries over to the case $|q|>\epsilon$ and we therefore skip the details here.
\end{proof}

\section{Proof of Theorem~\ref{thm:rel-tu-difference-strongc}}\label{sec:pf4}
The proof of Theorem~\ref{thm:rel-tu-difference-strongc} follows the same strategy as the proof of the analogous result for $T_l^1$ instead of $T_u^1$ in \cite[Theorem 1.1]{hainzl_boundary_2022}.
Since the proof only requires two non-trivial modifications, we do not reproduce the whole argument here, but only explain the necessary changes.
It should be noted that in \cite{hainzl_boundary_2022} the momenta in the relative and center of mass coordinates are scaled by a factor of $1/2$ compared to this paper, which becomes evident for instance when comparing the definition of $B_T$ in \cite[(2.2)]{hainzl_boundary_2022} with the one in \eqref{BT}.

Recall the definitions of $N_T(p,q)$ and $B_T(p,q)$ in \eqref{NT} and \eqref{BT}, respectively.
In this section, we shall write $N_{T,\mu}$ and $B_{T,\mu}$ to keep track of the value of $\mu$ occurring in $N_T$ and $B_T$ explicitly.

The first part of the argument in \cite{hainzl_boundary_2022} is to rephrase the statement about the temperatures $T_l^1$ and $T_c^0$ in terms of the corresponding Birman-Schwinger operators.
It turns out that the relative difference of the temperatures $T_l^1$ and $T_c^0$ vanishes in the strong coupling limit, if the lowest eigenvalues of the Birman-Schwinger operators at temperature $T=1$ and chemical potential $\mu=0$ agree.
This is the content of \cite[Lemma 5.1 and Lemma 5.2]{hainzl_boundary_2022}.
If we replace $T_l^1$ by $T_u^1$ and correspondingly $B_{T,\mu}$ by $N_{T,\mu}$, the arguments in the proofs of these two Lemmas remain valid, except for part (iv) of Lemma 5.1.
Let $N_{T,\mu}$ denote the operator on $L^2(\BR)$ with integral kernel $N_{T,\mu}(p,q)$.
In analogy to \cite[(5.3)]{hainzl_boundary_2022}, we need to show that
\begin{lemma}\label{lea:1d_strongc_1}
\begin{equation}
\lim_{\mu\to 0} \lVert N_{1,\mu}-N_{1,0}\rVert =0
\end{equation}
and
\begin{equation}
\lim_{\mu\to 0} \sup_p \left \vert \int_\BR (N_{1,\mu}(p,q)-N_{1,0}(p,q)) \dd q \right \vert=0.
\end{equation}
\end{lemma}
Before we prove this Lemma, we shall explain the other nontrivial change.
The second part of the argument in \cite{hainzl_boundary_2022} is to argue that the largest eigenvalues of the Birman-Schwinger operators corresponding to $T_l^1$ and $T_c^0$ at $T=1,\mu=0$ agree, which follows directly from the properties of the function $k(p,q)=\min \{ B_{1,0}(p,0),B_{1,0}(0,q)\}$ listed in \cite[Lemma 5.4]{hainzl_boundary_2022}.
Here, we define $k(p,q)=\min \{ N_{1,0}(p,0),N_{1,0}(0,q)\}$.
Since $N_{T,\mu}(p,0)=B_{T,\mu}(p,0)$, this coincides with the definition of $k$ in \cite{hainzl_boundary_2022} (up to scaling $p$ with $1/2$).
The only property listed in \cite[Lemma 5.4]{hainzl_boundary_2022} requiring a new proof, is the first inequality
\begin{equation}
N_{1,0}(p,q)\leq k(p,q).
\end{equation}
To prove this, observe that $x/\tanh(x)$ is convex and thus for all $x,y\in \BR$
\begin{equation}\label{eq:1d_strongc_1}
\frac{1}{2}\left(\frac{x}{\tanh(x)}+\frac{y}{\tanh(y)}\right) \geq \frac{\frac{1}{2}(x+y)}{\tanh(\frac{1}{2}(x+y))}
\end{equation}
Setting $x=(p+q)^2/2$ and $y=(p-q)^2/2$, and using that $x/\tanh x $ is monotonously increasing for $x\geq 0$, we obtain
\begin{equation}
\frac{1}{2} N_{1,0}(p,q)^{-1} \geq \frac{(p^2+q^2)/2}{\tanh((p^2+q^2)/2)} \geq \frac{1}{2} \max\left \{\frac{p^2}{\tanh(p^2/2)},\frac{q^2}{\tanh(q^2/2)}\right\}=\frac{1}{2} k(p,q)^{-1}
\end{equation}
Therefore, $N_{1,0}(p,q)\leq k(p,q)$.

It only remains to prove Lemma~\ref{lea:1d_strongc_1}.
\begin{proof}[Proof of Lemma~\ref{lea:1d_strongc_1}]
To show that $ \lVert N_{1,\mu}-N_{1,0}\rVert $ vanishes as $\mu\to 0$, we bound the operator norm by the Hilbert-Schmidt norm,
\[
 \lVert N_{1,\mu}-N_{1,0}\rVert ^2 \leq \int_{\BR^2} (N_{1,\mu}(p,q)-N_{1,0}(p,q))^2 \dd p \dd q.
\]
Since $N_{T,\mu}(p,q)\leq 1/2T$ and $N_{T,\mu}(p,q)\leq \frac{2}{|p^2+q^2-\mu|}$, there is a constant $c$ such that for all $\mu<1$ and $p,q \in \BR$
\[
N_{1,\mu}(p,q)\leq \frac{c}{p^2+q^2+1}.
\]
The claim thus follows from dominated convergence.

To show that
\begin{equation}
\sup_p \left \vert \int_\BR (N_{1,\mu}(p,q)-N_{1,0}(p,q)) \dd q \right \vert
\end{equation}
vanishes as $\mu\to 0$, we bound this expression above by
\begin{equation}\label{eq:1d_strongc_2}
\mu \sup_p \sup_{\nu\in[0,\mu]}\int_\BR \left \vert  \frac{\partial}{\partial \nu }N_{1,\nu}(p,q) \right \vert  \dd q.
\end{equation}
With the notation $f(x)=x/\tanh(x/2)$ we have $N_{1,\nu}(p,q)=2 (f((p+q)^2-\nu)+f((p-q)^2-\nu) )^{-1}$.
The derivative with respect to $\nu$ is given by
\begin{equation}
 \frac{\partial}{\partial \nu }N_{1,\nu}(p,q) = 2 \left(f((p+q)^2-\nu)+f((p-q)^2-\nu) \right)^{-2} (f'((p+q)^2-\nu)+f'((p-q)^2-\nu)),
\end{equation}
where
\begin{equation}
f'(x)= \frac{1}{\tanh(x/2)}-\frac{x/2}{\sinh^2(x/2)}.
\end{equation}
Using that $|f'(x)|<1$ for the second factor and \eqref{eq:1d_strongc_1} for the first term, we bound the derivative of $N_{1,\nu}$ by
\begin{equation}
 \left|\frac{\partial}{\partial \nu }N_{1,\nu}(p,q)\right| \leq \frac{1}{f(p^2+q^2-\nu)^2}
\end{equation}
To bound this further, we now restrict to $\nu\leq  1$ and use $f(x)\geq 2 \max\{1,x\}$, and then maximize over $p$ and obtain
\begin{equation}
 \left|\frac{\partial}{\partial \nu }N_{1,\nu}(p,q)\right| \leq \frac{1}{4}\chi_{p^2+q^2<2}+ \frac{1}{4}\frac{\chi_{p^2+q^2\geq 2}}{(p^2+q^2-1)^2}\leq  \frac{1}{4}\chi_{q^2<2}+\frac{c}{(q^2+1)^2}
\end{equation}
for some finite constant $c$.
Since this is integrable, the expression in \eqref{eq:1d_strongc_2} vanishes in the limit $\mu\to 0$.
\end{proof}

\paragraph*{Acknowledgments.}
I thank Marius Lemm and Andreas Deuchert for fruitful discussions and Robert Seiringer for valuable comments on the manuscript.
Funding by the Deutsche Forschungsgemeinschaft (DFG, German Research Foundation) – TRR 352 – Project-ID 470903074 is gratefully acknowledged.

\paragraph*{Conflict of interest.} The author has no conflicts to disclose.

\paragraph*{Data availability statement.}
Data sharing is not applicable to this article as no new data were created or analyzed in this study.

\bibliographystyle{abbrv}

\end{document}